\documentclass[12pt]{article}
\usepackage{subcaption}
\usepackage{changepage}
\usepackage{amsfonts}
\usepackage{geometry}
\usepackage{amsthm,amsmath,amssymb,amsfonts}
\usepackage[authoryear,longnamesfirst]{natbib}
\usepackage{booktabs}
\usepackage{graphicx}
\usepackage{nicefrac}
\usepackage{xr}
\usepackage{multirow}
\usepackage{relsize}
\usepackage{titling}
\usepackage[cmyk,x11names]{xcolor}
\usepackage[utf8x]{inputenc}
\usepackage{dsfont}
\usepackage{pgf,pgfplots,import}
\pgfplotsset{compat=newest}
\usepackage[normalem]{ulem}
\definecolor{accentcolor}{RGB}{50, 70, 150}
\usepackage{hyperref}
\hypersetup{pdfborderstyle={},allcolors=accentcolor,colorlinks}
\newenvironment{tablenotes}[1][Note]{\begin{minipage}[t]{\linewidth}\footnotesize{\itshape#1: }}{\end{minipage}}

\usepackage{bbm}
\usepackage[font=normalsize,justification=centering]{caption}
\DeclareCaptionLabelFormat{captionf}{\textsc{#1~#2}}
\usepackage{lscape} 
\usepackage{placeins}
\providecommand{\U}[1]{\protect\rule{.1in}{.1in}}
\usepackage[noamsmath]{kpfonts}
\usepackage{inconsolata}

\usepackage[T1]{fontenc}
\newtheorem*{proposition*}{Proposition}

\newtheorem{proposition}{Proposition}

\theoremstyle{definition}

\usepackage[capitalise]{cleveref}
\crefname{observation}{Observation}{Observations}
\crefname{assumption}{Assumption}{Assumptions}
\crefname{proposition}{Proposition}{Propositions}

\definecolor{ColorRed}{RGB}{238,26,28}
\definecolor{ColorBlue}{RGB}{55,126,184}
\definecolor{ColorGreen}{RGB}{77,215,74}
\colorlet{colav}{ColorRed!70!black}
\colorlet{colneu}{ColorBlue!70!black}
\colorlet{colseek}{ColorGreen!70!black}

\usepackage{setspace}
\def\Na{\mathbf{N}}

\def\ep{\varepsilon}
\def\one{\mathbf{1}}

\def\ta{\theta}

\def\la{\lambda}
\def\da{\delta}
\def\Re{\mathbf{R}} 
\def\Na{\mathbf{N}}

\def\os{\emptyset}

\newcommand{\df}[1]{\textit{#1}}

\newcommand{\abs}[1]{ \left | #1 \right | }

\newdimen\slantmathcorr
\def\oversl#1{
\setbox0=\hbox{$#1$}
\slantmathcorr=\wd0
\hskip 0.2\slantmathcorr \overline{\hbox to 0.8\wd0{%
\vphantom{\hbox{$#1$}}}}
\hskip-\wd0\hbox{$#1$}
}

\def\undersl#1{
\setbox0=\hbox{$#1$}
\slantmathcorr=\wd0
\underline{\hbox to 0.8\wd0{%
\vphantom{\hbox{$#1$}}}}
\hskip-0.8\wd0\hbox{$#1$}
}

\title{\textsc{Top of the Batch: Interviews and the Match}}

 \author{\small Federico Echenique\thanks{Division of the Humanities and Social Sciences, Caltech.} \quad \quad Ruy Gonz\'{a}lez\thanks{Caltech.} \quad \quad Alistair Wilson\thanks{Department of Economics, University of Pittsburgh.} \quad \quad Leeat Yariv\thanks{Department of Economics, Princeton University, CEPR, and NBER.}\ \thanks{Echenique gratefully acknowledges the support of NSF SES-1558757 and CNS-1518941. Yariv gratefully acknowledges the support of NSF grant SES-1629613.}}

\date{February, 2020} 

\renewenvironment{abstract}
  {\small
  {\bfseries\noindent{\abstractname}\par\nobreak\smallskip}}

\date{December, 2020}

\begin{document}
\title{Top of the Batch: Interviews and the Match}
 \maketitle
\begin{abstract}
Most doctors in the NRMP match with one of their most-preferred internship programs. However, surveys indicate doctors' preferences are similar, suggesting a puzzle: how can so many doctors match with their top choices when positions are scarce? We provide one possible explanation. We show that the patterns in the NRMP data may be an artifact of the interview process that precedes the match. Our study highlights the importance of understanding market interactions occurring before and after a matching clearinghouse, and casts doubts on analyses of clearinghouses that take reported preferences at face value.
\end{abstract}




\thispagestyle{empty}

\vskip 0.2in

\noindent \textbf{Key words:} NRMP, Deferred acceptance, Interviews, First-rank matches

\vskip 0.2in

\noindent \textbf{JEL:}
C78, 	
D47,  
J44	 

\newpage \renewcommand{\thefootnote}{\arabic{footnote}} \pagebreak %
\setcounter{page}{1}
\section{Introduction}

The National Resident Matching Program (NRMP) has matched millions of doctors to residency programs across the United States. In 2020 alone, 45,000 active applicants matched to over 37,000 positions. Match results reported by the NRMP for 2020 suggest comforting news for doctors: 46.3\% of freshly-minted MDs from US schools were matched to their first-ranked choice, while 71.1\% were matched to one of their top-three choices. The most-recent year's figures are by no means an aberration. The fraction of applicants matched to their first-ranked choice has been at least as high over the past two decades. We suggest these surprising figures should not be taken at face value. In particular, we show that interactions outside of the main match---through the interview process that precedes it----may be at least as important as the matching protocol itself. 

Why should a very large fraction of doctors matching to their top-ranked residencies be surprising? The algorithm governing the NRMP match implements a stable matching over the reported preferences. If applicants report similar preferences, only a few applicants can get their most-preferred option. For example, suppose we wish to match 100 prospective residents to 100 positions. Common preferences on both sides (an assortative market) yield an outcome where just \textit{1\% } of doctors are matched to their first-ranked program. As we show, even a small common component in doctors' preferences implies relatively few matches to top-ranked hospitals. 

One explanation for the NRMP outcome data is that applicants' preferences are diametrically opposed, with a handful of applicants ranking each position as their top outcome. This stands in the face of survey data and preference estimations suggesting important preference commonalities \citep[see][]{rees2018suboptimal, agarwal2015match}. Another explanation might be that preferences are independent, or even somewhat correlated, across participants but that each doctor and hospital consider only $k$ of their top partners as acceptable, as in \cite{10.1145/2656202}, and submit those preferences truthfully. Matched participants would then have to receive one of their top-$k$ partners. As we show, this explanation too has shortcomings. First, it does not explain the relative prevalence of matches with the first-ranked partner. Second, for small $k$---which is arguably the case in the NRMP, where doctors commonly rank fewer than $20$ programs---many applicants remain unmatched under truncation to the top-$k$ partners \citep[see also][]{arnosti2015,beyhaghitwo,lee2016incentive}.   

We propose another story. Prior to the NRMP, applicants interview with hospitals. The  determination of who interviews with whom is decentralized with two important features.  First, interviewing is costly, and capacities are limited.  Second, hospitals and doctors submit rankings to the NRMP only for those \textit{they interviewed with}.\footnote{The 2019 NRMP Applicant Survey (available from \url{nrmp.org}) reports on four types of median respondents in 21 specialties (anesthesiology, pediatrics, etc.). Of the 84 medians reported, 63 have perfectly coincident numbers for interviews attended and programs ranked, where 81 are $\pm 1$.
}

We assume that hospitals and prospective residents' preferences are decomposable into common and idiosyncratic components. For hospitals, the common component can reflect doctors' academic performance and test scores \citep{agarwal2015match}. For doctors, it can reflect hospital rankings, quality of life in the local area, etc. In contrast, the idiosyncratic component reflects match-specific values. Assuming this preference form, we consider a pre-match interview-selection process. Each hospital has a maximum number of interview slots, $k$, while each candidate has a limit on the number of interviews they can attend, $k^\prime$. The decentralized interview outcome is then modeled as a stable many-to-many matching under the ($k$,$k^\prime$) capacity constraints. At the centralized matching stage, only interview partners' ranks are reported, which we refer to as the ``interview-truncated'' preferences. 

The truncation induced by the interview process necessarily narrows agents' original preferences. Nonetheless, since hospitals' and doctors' preferences are linked through stability of the interview process, a large fraction of prospective doctors still end up matched. Moreover, \textit{reported} ranks for match outcomes are much higher than in the untruncated preferences. 

The presence of a common component in prospective residents' preferences is crucial for this conclusion. In particular, we show that with sufficient disagreement in doctors' preferences, interviews may cause matched partners' reported rank to go down, not up. While perfect agreement among the doctors over hospital rankings implies that interviews lead to inflated rankings for matched programs, this obviously represents an extreme.\footnote{A related idea appears in \cite{beyhaghitwo}, who show that interviews may increase the size of a match. See also \cite{kadam2015interviewing}.} Our main theoretical finding is that in large markets, an \textit{arbitrarily weak} common component is sufficient for interviews to generate the pattern of high-reported ranks for match partners.

As our most-general result is asymptotic, we complement it with simulations at more moderate market sizes. Not only do our simulation results strongly mirror the NRMP outcome reports (unmatched fraction, distribution of submitted rankings), they also provide a strong link with one of the other main findings in the literature, that of small-cores in \citet{roth1999redesign}.\footnote{In an environment with fully idiosyncratic preferences, \cite{ashlagi2017unbalanced} show that imbalanced markets lead to high reported ranks for the short side of the market---at the aggregate level for the NRMP, the hospital side.}

The idea that doctors' reports in the residency match may not reflect true preferences is certainly present in other work. \cite{10.1257/aer.p20171027} survey evidence of misreports in the NRMP, suggesting four possible explanations: proposers' failure to identify the dominant strategy, mistrust in the mechanism, non-classical utility,
and self-selection. The last of these is closest to the mechanism in our paper. In this vein \cite{CHEN201959} consider school-choice problems where students ``self select'' by only ranking schools they believe will plausibly admit them,  showing evidence for this self-selection in Mexican high-school applications. While doctors and hospitals only ranking those they interview with is a manifestation of self-selection, our theoretical analysis offers a constructive process to shed light on this process and its impact on outcomes.\footnote{\cite{LeeSchwartzInterviews} also consider an interview process that precedes a centralized match. In their setting, workers are fully informed of their preferences, while firms view workers symmetrically at the outset and use costly interviews to infer their own preferences. 
In the NRMP context, \cite{rees2018suboptimal} uses surveys to illustrate doctors' significant ``misreporting'' in the match, while \cite{Rees-Jones11471} uses an online experiment with post-match medical students where 23 percent misrepresent their preferences in an incentivized NRMP-like matching task.}

Our results have important implications for the NRMP, and the matching literature more broadly. Doctors participating in the deferred-acceptance algorithm underlying the match have incentives to truthfully report preferences \citep{roth1999redesign}. Traditionally, economists have viewed the NRMP as an ideal case-study in strategy-proof design. Our findings suggest that because reported preferences in the NRMP are filtered through the interview stage, they should be interpreted with caution. In particular, reported high-rank matches cannot be read literally, and any conclusions drawn about welfare using estimated preferences from the match itself are suspect. This message is particularly stark given that our paper ignores strategic effects at interviews.\footnote{See \cite{beyhaghi2017effect} for an analysis of the strategic implications of interviews.} 

\section{Setting Up the Puzzle}
We first argue that standard preference assumptions on the match process are at odds with observed outcomes in the NRMP data. We focus on two key measures. First, across two decades of annual matches, approximately one half of all matched residents obtain their first-ranked outcome.\footnote{See Figure B.1 in the Online Appendix for details.} This holds not only for matched US MD seniors (residents graduating with an MD from a US medical school), but also for matched independent applicants (those from DO-granting schools, or based outside of the US) who match at lower rates. Second, \citeauthor{roth1999redesign} demonstrate that reported rank-order lists exhibit small cores. That is, in a shift from doctor-proposing to hospital-proposing deferred acceptance (DA), they show that only $0.1\%$ of doctors have different outcomes. Taking the reported preferences at face value, this means that the vast majority (99.9\%) have a unique stable-match partner: a small core.

To gain intuition for which preferences can generate these patterns, consider a two-sided matching market with $N=100$ participants on each side. For each doctor $d$ and hospital $h$, the respective cardinal-match utilities are given by:
\[u_d(h)=\lambda^D\cdot c_h + (1-\lambda^D)\cdot\eta_{d,h}\text{ and } u_h(d)=\lambda^H\cdot c_d + (1-\lambda^H)\cdot\eta_{h,d}. \]
For the $c_h$ ($c_d$) terms we draw i.i.d. $\mathcal{N}(0,1)$ random variables to represent common-utility components of matching with each hospital (doctor). The $\eta$ terms represent idiosyncratic terms, again i.i.d. $\mathcal{N}(0,1)$ random variables. Finally, $\lambda^H$ and  $\lambda^D$  represent the relative weights on the common and idiosyncratic utility components for each side of the market.

We form analogs to the two features of NRMP outcomes using 500 DA simulations with the above preference assumptions.  Figure~\ref{fig:GSsim}(a) indicates the simulated fraction of doctors matched to their first-ranked program under doctor-proposing DA, while Figure~\ref{fig:GSsim}(b) depicts the fraction with the same partner under both doctor- and hospital-proposing DA. On the horizontal axis we vary the hospitals' common-component weight, where each plotted curve varies the doctors' common-component weight.

\begin{figure}[tb]
\begin{center}
\begin{subfigure}{0.49\textwidth}
\includegraphics[width=0.99\textwidth]{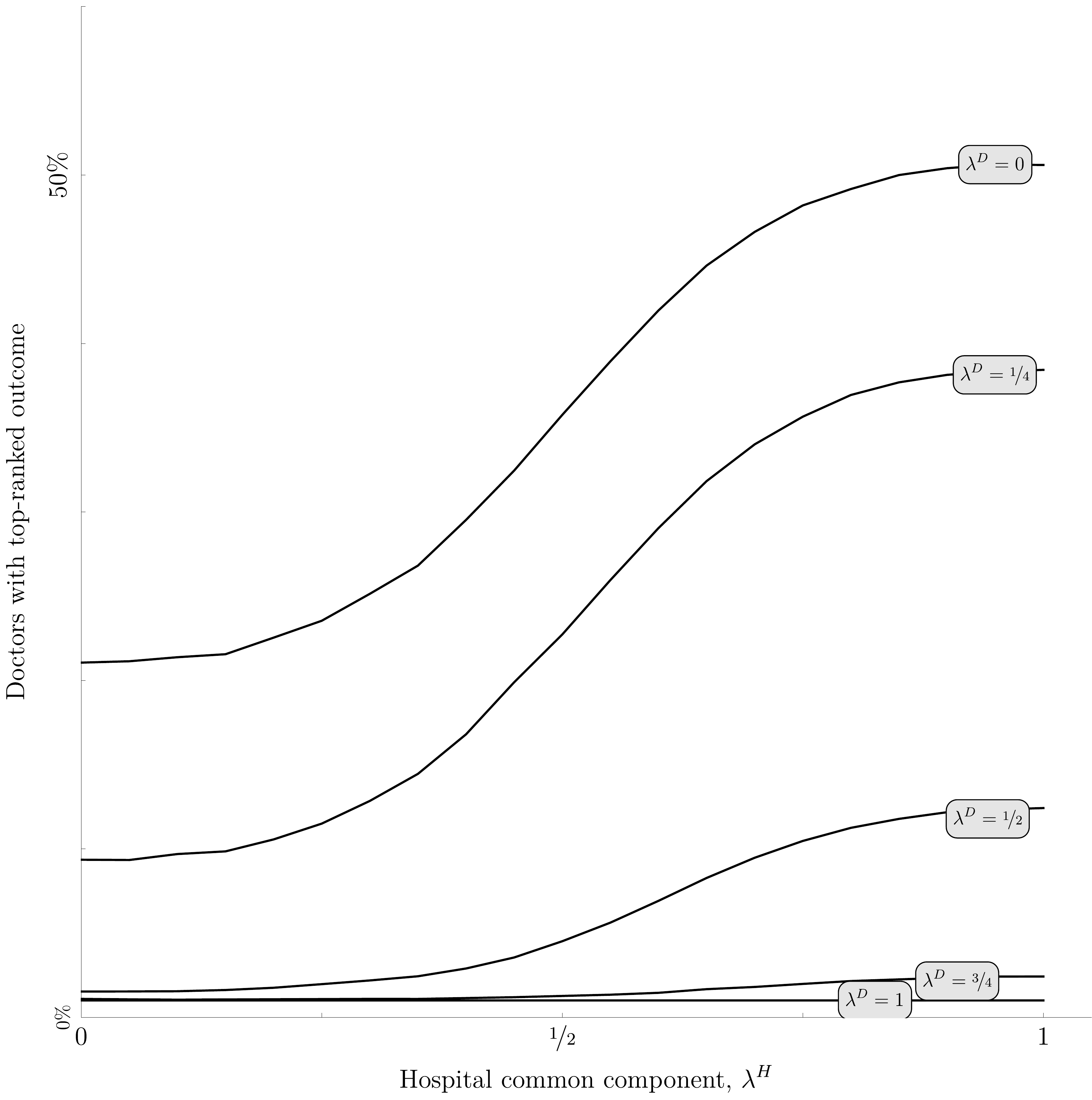}
\caption{Top-ranked matches}
\end{subfigure}\begin{subfigure}{0.49\textwidth}
\includegraphics[width=0.99\textwidth]{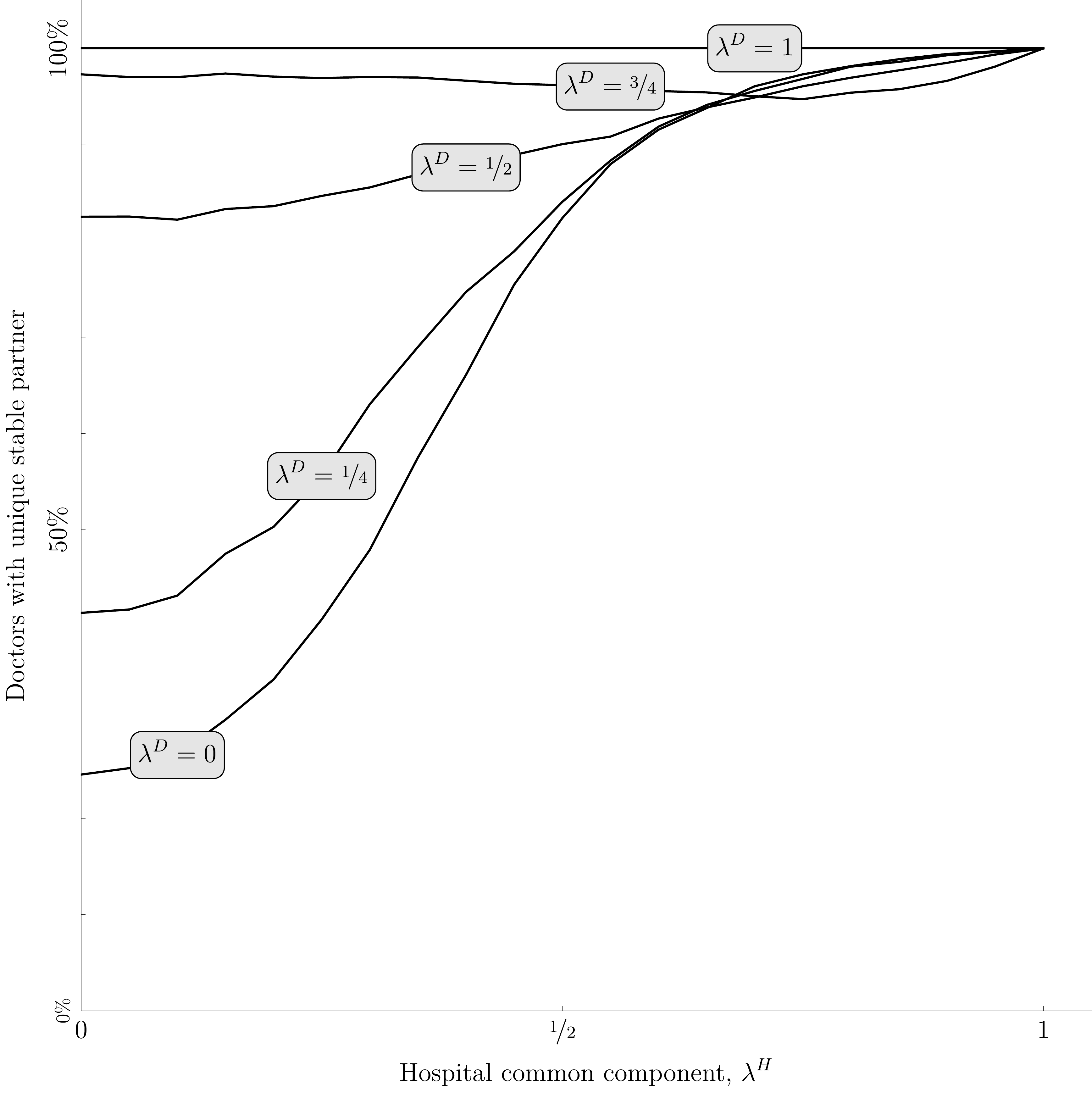}
\caption{Unique stable-partner}
\end{subfigure}
\caption{Simulated DA outcomes\label{fig:GSsim}}
\end{center}
\end{figure}

Having almost half of the doctors attain their first-ranked match, and almost all with a unique stable-match partner, is attained at an extreme for the preference weights. This happens when hospitals' preferences are driven almost entirely by the common component ($\lambda^H$ close to one) and doctors' preferences by the idiosyncratic component ($\lambda^D$ close to zero). 

Hospitals having a strong common component is consistent with NRMP survey data.\footnote{From the Director's Survey,  hospitals place substantial weight on features such as test scores, recommendation letters, etc.}. However, the requirement that doctors' preferences are almost completely idiosyncratic contradicts ample survey evidence. The NRMP's 2019 post-match resident survey suggests that common-value components (``reputation of program,'' having an ``academic medical center program,'' as well as quality of the residents, faculty, and educational curriculum) are cited at similar frequencies to idiosyncratic ones (``perceived goodness of fit" and ``geographic location'') as reasons for ranking programs.



We next show theoretically how the interview process can help reconcile these observations. We then use simulations to connect our framework with the empirical regularities observed in the NRMP.


\section{Theory}\label{sec:themodel}
Our model is a variant of the standard two-sided matching model \citep[see, for example,][]{rothsotomayor1990}, with an added interview stage. 

\subsection{Basic Definitions}\label{sec:defns}
A \df{market} is a triple $(H,D,U)$, where: $H$ is a finite set of \df{hospitals}; $D$ is a finite set of \df{doctors}; and $U=( (u_d)_{d\in D}, (u_h)_{h\in H})$ is \df{utility function} profile (with $u_d:H\cup\{d\}\rightarrow \Re$ and 
$u_h:D\cup\{h\}\rightarrow \Re$ for each $d$ and $h$).

A utility $u_a$ induces an ordinal preference $\succeq_a$ over the relevant set of alternatives, where we assume throughout that 
the resulting ordinal preferences are \df{strict}. The \df{rank-order} of $b$ in $u_a$ is one plus the number of
$b'$ with $u_a(b')>u_a(b)$---so that a lower rank-order indicates a better ordinal outcome/higher ranking. In particular, agent $a$'s most-preferred match partner has rank-order 1. An agent $b$ is \df{unacceptable} for $a$ if $u_a(a)> u_a(b)$.

A \df{matching} is a function $\mu:H\cup D\rightarrow H\cup D$, with the properties that $\mu(h)\in D\cup\{h\}$,  $\mu(d)\in H\cup\{d\}$, and $\mu(d) = h$ \emph{iff} $\mu(h)=d$. A matching $\mu$ is \df{stable} for a market $(H,D,U)$ if $u_a(\mu(a))\geq u_a(a)$ for all $a\in D\cup H$, and there is no $(d,h)\in D\times H$ with $u_d(h)>u_d(\mu(d))$ and $u_h(d)>u_h(\mu(h))$.

A \df{many-to-many} matching is a function $\mu: H\cup D \rightarrow 2^{H\cup D}$ with the properties that $\mu(d)\subseteq H$,
$\mu(h)\subseteq D$, and $h\in\mu(d)$ \emph{iff} $d\in\mu(h)$. When an agent $a$ is unassigned, we have $\mu(a)=\os$. Given a pair of positive integers $(k,k')$, a many-to-many matching $\mu$ is \df{pairwise stable} for $(k,k')$ if
\begin{itemize}
  \item $\abs{\mu(d)} \leq k$ and there is no $h\in \mu(d)$ with $u_d(h)<u_d(d)$;
\item $\abs{\mu(h)} \leq k'$ and there is no $d\in \mu(h)$ with $u_h(d)<u_h(h)$;
\item 
There is no $(h,d)$ such that $d\notin\mu(h)$ and any one of the following:
  \begin{itemize}
    \item $u_d(h)> u_d(h')$ and $u_h(d)> u_h(d')$ for some $(h',d')\in
      \mu(d)\times \mu(h)$;
      \item $u_d(h)> u_d(h')$, $u_h(d)> u_h(h)$, and $\abs{\mu(h)}<k'$ for some $h'\in\mu(d)$;
      \item $u_d(h)> u_d(d)$, $u_h(d)> u_h(d')$, and $\abs{\mu(d)}<k$ for some $d'\in \mu(h)$.
    \end{itemize}
\end{itemize}

\subsection{Interview Schedules}\label{sec:interviews}
In our model, doctors and hospitals first schedule interviews and then participate in the match. 

An \df{interview schedule} is a many-to-many matching. Given a pair of integers $(k,k')$, a \df{$(k,k')$-constrained interview schedule} is a many-to-many matching $\mu$ with $\abs{\mu(d)} \leq k$  and $\abs{\mu(h)} \leq k'$ for all $d$ and~$h$. Each doctor can interview with at most $k$ hospitals, and each hospital can interview at most $k'$ doctors. 

Given an interview schedule $\mu$, agents' \df{interview-truncated} preferences are determined by setting $u_a(b)<u_a(a)$ for all $b\notin \mu(a)$. That is, interview-truncated preferences rank all interviewed agents as in the original preferences, and set all other agents as unacceptable. 

The timing in our model is then: (i) An interview schedule is determined as the doctor-optimal many-to-many $(k,k')$-stable matching;\footnote{Arguably, the doctor-optimal stable matching at the interview stage yields a smaller difference between reported and actual ranks than other selections of stable matchings.} (ii) Doctors and hospitals report their interview-truncated preferences as inputs into doctor-proposing DA. This process' outcome is therefore the doctor-optimal stable matching on the interview-truncated preferences. We term this two-step process Int-DA: the \textbf{Int}erview process followed by \textbf{D}eferred \textbf{A}cceptance.

A doctor-optimal interview schedule can be found algorithmically using the ``T-algorithm'' \citep[see][]{blair1988lattice,fleiner2003fixed,TE2006233}. We assume it is the result of a decentralized interview scheduling process. One may imagine several reasons why an interview schedule might be unstable. Our focus is on the tension between a ``pure'' application of DA, and one that is preceded by interviews. Assuming a stable outcome at the interview stage provides us with a simple, tractable model.\footnote{In one-to-one matching markets, experimental evidence suggests decentralized interactions yield stable outcomes at high rates, see \cite{echenique2012experimental}. \cite{melcherashlagiwapnir2019} propose a stable-matching algorithm for internship interviews. For more on the theory of many-to-many matching, see \cite{sotomayor1999three,konishi2006credible}.}

We denote the final matching from Int-DA as $\mu^I$. We will compare the Int-DA matching to that obtained from the doctor-proposing DA algorithm using agents' original preferences,  $\mu^{DA}$.

\subsection{Interviews can increase rank-orders}
In general, interviews alone cannot explain the findings in the data: Int-DA does \emph{not} necessarily yield better-ranked partners in submitted preferences.





As a simple example, consider a matching market with three doctors, $\{d_1,d_2,d_3\}$, and four hospitals, $\{h_1,h_2,h_3,h_4\}$ (it is easy to concoct slightly more complicated examples with the same number of doctors and hospitals). Hospitals' preferences are common: they all prefer $d_1$ to $d_2$, $d_2$ to $d_3$, and $d_3$ to staying unmatched. Doctors' rank all hospitals as acceptable, with preferences given by (first to last):
    \begin{center}
        \begin{tabular}{rl}
        $d_1$: & $h_1$, $h_3$, $h_2$, $h_4$;\\
        $d_2$: & $h_2$, $h_3$, $h_1$, $h_4$;\\
        $d_3$: & $h_3$, $h_1$, $h_4$, $h_2$.\\
        \end{tabular}
    \end{center}

\bigskip

Under DA, $d_i$ matches to $h_i$. So the rank-order of $d_3$'s match is $1$. 

Suppose interview constraints are $k=k'=2$. All doctors want to interview with $h_3$, but only $d_1$ and $d_2$ are able to. The resulting interview schedule is: $d_1$ with $h_1$ and $h_3$; $d_2$ with $h_2$ and $h_3$; and $d_3$ with $h_1$ and $h_4$.

Given the interview-truncated preferences, $d_i$ matches with $h_i$ for $i=1,2$, but $h_3$ is matched with $d_4$. Thus, the Int-DA rank-order of $d_3$'s match is $2$. Hence, the rank of $d_3$'s partner in the presence of interviews is strictly worse than the rank of her partner under DA when no interviews take place. In fact, the outcome under Int-DA is unstable for the original preferences.

In this example, there is substantial disagreement between doctors' preferences. Indeed, there are no pairwise comparisons of hospitals $\{h_1,h_2,h_3\}$ on which doctors agree. In what follows, we show that some agreement on hospitals' rankings rules out such examples, and interviews can explain observed high match ranks.

\subsection{Interviews with Common Preference Components}\label{sec:theory}
Our discussion of the NRMP data emphasized the role of common components in doctors' and hospitals' preferences. Our first theoretical result (Proposition~\ref{prop:aligned}) confirms that, indeed, if doctors agree on hospitals' ranking, interviews improve observed match ranks in the succeeding clearinghouse. Our second result (Proposition~\ref{propLM}) shows that, as long as there is a common-value component in agents' preferences, \emph{however small}, the message of our first result holds in large markets. Finally, we illustrate convergence rates for the large-market result (Proposition~\ref{prop:convrates}).

\subsubsection*{Aligned Preferences}\label{sec:alignedpreferences}

We start with the extreme case where doctors' preferences are common.

\begin{proposition}\label{prop:aligned} Suppose $k=k'$ and that doctors' preferences are identical. For any doctor $d$, the rank-order of $\mu^I(d)$ in her interview-truncated preference is always weakly lower than the rank-order of $\mu^{DA}(d)$ in her actual preference $\succeq_d$.  
\end{proposition}

The proof appears in the Online Appendix. Intuitively, when doctors' preferences are common, only one of the doctors under DA is matched to the highest-ranked hospital, one to the second-highest, etc. In particular, $n-k$ doctors are matched to a hospital ranked below their top $k$. In contrast, interviews allow for presorting of doctors to hospitals they have a chance of matching with. Interviews also limit how low a matched hospital can be ranked in the reported preferences: it can never be lower than $k$.    

The proposition assumes $k=k'$, mainly for expository reasons. In our main result below we allow for the two bounds to differ. We also show that simple truncation of preferences submitted to DA, absent interviews, cannot explain the gamut of stylized facts suggested by NRMP data.

\subsubsection*{Large Markets}\label{sec:largemkts}
We expand the model to account for market size, and for randomly generated preferences. For each $n$, let $(D_n,H_n,U_n)$ denote a market, where $D_n =\{d_1,\ldots,d_n\}$, $H_n = \{h_1,\ldots,h_n\}$ and each utility function is randomly drawn with a common-value and idiosyncratic component. As before, suppose that
\[
    u^n_d(h)  = \la^D c_h +  (1-\la^D) \eta_{d,h} \text{ and } \ u^n_h(d)  =  \la^H c_d + (1- \la^H) \eta_{h,d},
\]
for all $d\in D_n$ and $h\in H_n$, where $\la^D,\la^H\in (0,1)$. Suppose, moreover, that $u_a^n(a)=0$. 
The common-value components $c_h$ and $c_d$ are crucial for our results, but need not dominate doctors' utilities, so $\la^D,\la^H>0$ can be arbitrarily small.

Suppose that $c_h$, $c_d$, $\eta_{d,h}$ and $\eta_{h,d}$ are all drawn from an absolutely continuous distribution with support $[0,1]$.\footnote{Any continuous distribution with strictly positive density and support on the positive reals suffices.} Let $\mu^{I}_n$ denote the matching resulting from the Int-DA process in the $n$-sized market, and  $\mu^{DA}_n$ the corresponding outcome of the doctor-proposing DA; these matchings are random and depend on the realized utilities, omitting the explicit dependence on $(c,\eta)$.

The Int-DA procedure determines a matching $\mu^I_n$ by choosing a $(k_n,k'_n)$-constrained interview schedule $\hat\mu$ as the doctor-optimal many-to-many stable matching, followed by the doctor-proposing DA using the induced preferences.

\begin{proposition} \label{propLM}
Suppose that $\limsup k_n/n < 1$ and
let $\ep,\ta > 0$. The probability of the following event converges to 1 as $n\rightarrow \infty$: For a fraction of at least
$1-\ta$ of doctors $d\in D_n$, the rank-order of $\mu^I_n(d)$ in $d$'s interview-truncated preference is strictly below the rank-order of $\mu^{DA}_n(d)$ in $d$'s actual preference $\succeq_d$.
\end{proposition}

A fully formal statement, and 
the proof of  Proposition~\ref{propLM}, appear in the Appendix. 

The idea underlying the proposition is simple. Consider DA and let $\ep>0$. By \cite{lee2016incentive}, when $n$ is large, with high probability, the set $A_n(\ep,(c,\eta))$ of doctors that are within $\ep$ of their ``target'' assortative utility in DA account for at least $1-\ta/2$ of all doctors. Let $B(c_d,n)$ be the event that fewer than $k_n$ hospitals give doctor $d$ a utility greater than $d$'s target utility. We denote by $\beta_n$ the probability that a fraction of at least $\theta/2$ doctors have a ``small'' number (at most $k_n$) of hospitals above their target utility. We show that for $n$ large enough, $\beta_n <\pi/2$, and by \cite{lee2016incentive}, $P \left(\frac{1}{n} \abs{A_n(\ep,(c,\eta))} \geq 1-\ta/2\right) > 1-\pi/2$. Thus, the event that $B(c_d,n)$ is false for a fraction $\geq 1-\ta/2$ of doctors {\em and} $\frac{1}{n} \abs{A_n(\ep,(c,\ep))} \geq  1-\ta/2$, has probability $\geq (1-\pi/2 ) + (1-\pi/2 ) -1 = 1-\pi$. At the intersection of these conditions, for a fraction $\geq (1-\ta/2) + (1-\ta/2) -1 = 1-\ta$ of $d\in D_n$, we have that $B(c_d,n)$ is false and $d\in A_n(\ep)$. Hence, for a fraction $\geq 1-\ta$ of $d\in D_n$ there are more than $k_n$ hospitals above their target utility, and they are within $\ep$ of their target utilities. 

Finally, we note that convergence rates for the large-market result in Proposition~\ref{propLM} are modest, with (poly-)logarithmic or polynomial growth in the relevant ``approximation guarantees'' $\ta$ and $\pi$. In words, the market size needed for Proposition~\ref{propLM} does not grow too quickly with the approximation guarantees. This message complements the simulations in Section~\ref{sec:simulations}, which assume (arguably) realistic market sizes, and can be formalized as follows (detailed proof appears in the Online Appendix):
\begin{proposition}\label{prop:convrates}
The statement in Proposition~\ref{propLM} holds for $n=\Theta((\ln(1/\pi))^4)$ as $\pi\rightarrow 0$, and $n=\Theta((1/\ta)^4)$ as $\ta\rightarrow 0$.
\end{proposition}

\section{Simulations}\label{sec:simulations}
Our theoretical findings raise three important questions. The first regards market size. Proposition \ref{propLM} is asymptotic, and it is natural to consider  whether interviews matter for smaller, more realistic, market sizes. The second question regards unmatched agents. One might worry that interview-truncated preferences give rise to large numbers of unmatched participants, beyond those observed in the NRMP. The final question regards stability.
As our example in Section 3 makes clear, even though both interview selection and DA separately produce stable outcomes, their sequential application does not guarantee stability. Ideally, the difference between outcomes under DA and the interview-truncated DA procedure would be small.

We address these questions using numerical simulations at two market sizes: a small market of $N=50$, and a medium market of $N=1,700$.\footnote{While the 2020 NRMP had 37,256 positions listed, the match breaks down into a number of specialty sub-markets. For the 2019 NRMP outcome report (Table 13) the specialties vary in size from 22 positions for Pediatrics/Medical Genetics (the NRMP only provide data for specialties with more than 20 total positions) to 9,127 for Internal Medicine. The 20th and 80th percentiles across the listed sub-markets in 2019 have 37 and 1,740 positions listed, mirroring our chosen simulation sizes.}

In our simulations, we apply a weight of either $\tfrac{1}{4}$, $\tfrac{1}{2}$, or $\tfrac{3}{4}$ to the common components of both hospitals and doctors (so $\lambda^D=\lambda^H$).\footnote{In all simulations we use normally distributed common and idiosyncratic draws to derive cardinal preferences $u_d(h)$ and $u_h(d)$, per Section 2.} We conduct 340 simulations for the $N=50$ markets, and 10 simulations for $N=1,700$, leading to outcome information on 17,000 market participants at each market-size--preference-weight pair, $(N,\lambda)$.\footnote{All reported figures are averages across the simulations and doctors.} In each simulated market we first draw and fix market-wide preferences, and then calculate the following match outcomes:\renewcommand{\labelenumi}{(\roman{enumi})}
\begin{itemize}
    \item Doctor- and hospital-proposing deferred acceptance (DA).
    \item The stable interview allocation with $k=k^\prime=5$ slots per position, followed by both doctor- and hospital-proposing deferred acceptance on the interview-truncated preferences.
    \item Doctor- and hospital-proposing deferred acceptance on preferences truncated to the $k=5$ top-ranked options(Tr-DA).
\end{itemize}

\begin{table}[p]
    \centering
    \caption{Simulation Outcomes}
    \label{table:SimOutcomes}
    \begin{tabular}{cccccccc}
    \toprule
    & \multicolumn{3}{c}{$N=50$} &  & \multicolumn{3}{c}{$N=1700$}\\ \cmidrule{2-4}\cmidrule{6-8}
     & $\lambda=\nicefrac{1}{4}$ & $\lambda=\nicefrac{1}{2}$ & $\lambda=\nicefrac{3}{4}$ &  & $\lambda=\nicefrac{1}{4}$ & $\lambda=\nicefrac{1}{2}$ & $\lambda=\nicefrac{3}{4}$\\
     \midrule
      \multicolumn{7}{l}{\textbf{Panel A: Matching outcomes}}\\ 
& \multicolumn{7}{c}{\emph{Unmatched} \hspace{0.25in}$\left[\text{DA: }0.0\%,\text{NRMP: }5.4\%\right]^\dagger$ } \\ \cmidrule{2-8}
     Int-DA & 5.7\% & 5.3\% & 4.4\% &  & 6.0\% & 5.8\% & 5.5\%\\ 
    Tr-DA & 14.6\% & 39.9\% & 71.7\% &  & 24.2\% & 68.9\% & 95.4\%\\ 
    & \multicolumn{7}{c}{\emph{First-ranked program}\hspace{0.25in} $\left[\text{NRMP: }48.1\%\right]^\dagger$ } \\\cmidrule{2-8}
    DA & 16.1\% & 7.2\% & 3.6\% &  & 2.9\% & 0.6\% & 0.2\%\\
      Int-DA & 42.8\% & 39.8\% & 34.6\% &  & 43.0\% & 41.5\% & 40.6\%\\ 
     Tr-DA & 31.1\% & 12.2\% & 3.8\% &  & 24.1\% & 4.7\% & 0.3\%\\
     & \multicolumn{7}{c}{\emph{Top-three--ranked program match}\hspace{0.25in}  
     $\left[\text{NRMP: }73.2\%\right]^\dagger$ } \\\cmidrule{2-8}
    DA & 41.3\% & 21.3\% & 10.2\% &  & 8.3\% & 2.0\% & 0.6\%\\
      Int-DA & 81.9\% & 81.4\% & 80.7\% &  & 81.7\% & 81.3\% & 81.2\%\\ 
     Tr-DA & 68.5\% & 38.6\% & 12.9\% &  & 57.5\% & 17.6\% & 1.7\%\\
   \midrule
    \multicolumn{7}{l}{\textbf{Panel B: Core size, similarity to DA, and stability}}\\
     & \multicolumn{7}{c}{ \emph{Same partner under proposer change} \hspace{0.05in}$\left[\text{NRMP: }99.9\%\right]^\ddagger$ }  \\ \cmidrule{2-8}
    DA & 60.9\% & 88.7\% & 93.9\% &  & 43.7\% & 95.0\% & 98.9\%\\
      Int-DA & 98.4\% & 98.6\% & 97.6\% &  & 99.9\% & 99.9\% & 99.9\%\\ 
      & \multicolumn{7}{c}{\emph{Identical partner to DA}} \\ \cmidrule{2-8}
       Int-DA  & 74.0\% & 80.8\% & 82.3\% &  & 73.4\% & 78.1\% & 77.0\%\\ 
   &\multicolumn{7}{c}{\emph{Proportion blocking programs in Int-DA}} \\\cmidrule{2-8}
     Matched & 0.5\% & 0.8\% & 1.3\%  & & 0.1\%  & 0.5\% & 1.7\% \\
     Unmatched & 16.4\% & 20.3\% & 33.3\% &  & 10.0\% & 15.4\% & 32.6\% \\
\bottomrule
\end{tabular}
\begin{tablenotes}
$\dagger$--Average for US MD Seniors in 2016--20. Source:   \emph{Results and Data: 2020 Main Residency Match}, Table 15, available from nrmp.org.
$\ddagger$--Figure reported for main NRMP match in \citet{roth1999redesign}. Smaller thoracic surgery market ($N\simeq120$) has a 99.6 percent unique match for five reported years in 1991--96 (\emph{ibid}, tables 1 and 3). 
\end{tablenotes}
\end{table}

Table~\ref{table:SimOutcomes} provides outcomes from our simulations across the six $(N,\lambda)$ parameter pairs. Simulations with Tr-DA were added to distinguish the pure effect of truncation from the interview process our paper focuses on. The first panel in the table provides three characteristics of the match outcome: (i) the fraction of \emph{unmatched} participants; (ii) the fraction of doctors matched to their \emph{first-ranked program}; and (iii) the proportion of doctors matched to a \emph{top-three--ranked program}.

Because our simulated markets have the same participant volume on each side, with all possible matches acceptable, the benchmark for DA with full preferences predicts no unmatched doctors. In contrast, the NRMP data indicates that 5.8\% of US seniors are unmatched. The first result from our simulations in Table~\ref{table:SimOutcomes} illustrates that the two-stage Int-DA process leads to a similar unmatched rate as the NRMP. Doctors in our simulations are unmatched after the Int-DA process at a 5.5\% rate. Moreover, this proportion does not change substantially with either market size or the common weight. In contrast, a direct truncation to the top-five participants on the other side leads to substantially more unmatched participants. Moreover, the unmatched rate grows sharply with increases to $N$ and $\lambda$.

The next pair of results from the  Int-DA simulations again match the NRMP data: a large fraction of doctors are matched to top-ranked hospitals. Looking to NRMP data from the past five years, 48\% (73\%) of US MD Seniors are matched to their first-ranked (top-three--ranked) program. The Int-DA simulations again indicate similarly-sized effects to the observed NRMP figures, at 40\% (81\%).\footnote{The Int-DA fraction matched to their first-ranked program does increase slightly as we increase $N$, and decreases slightly as we increase $\lambda$.} In contrast, the pure DA algorithm on the full preference lists implies substantially lower rates of top-ranked outcomes, particularly  in larger markets and as the common weight increases.

In the second panel of Table \ref{table:SimOutcomes} we turn to other observed match outcomes.  These outcomes are not part of our explanation of reduced match ranks, but they serve to evaluate the empirical relevance of our interviews model. The first outcome is motivated by \citeauthor{roth1999redesign}'s \citeyearpar{roth1999redesign} finding that NRMP data exhibit small cores. Using NRMP ranking data from the 1990s, they examine the change in outcomes moving from the doctor- to the hospital-proposing DA. They find that 99.9\% of doctors receive the same outcome---implying a unique stable partner. In the \emph{same partner under proposer change} rows we mirror this exercise. Our DA simulations get close to the NRMP figure only in the larger markets with a heavy weight on the common component. While most participants across each of the simulations do have a unique stable partner, the minority with multiple partners are at least an order of magnitude larger than in \citet{roth1999redesign}. However, changing the proposing side over the interview-truncated rankings from Int-DA indicates much-closer effects to the NRMP field study. Indeed, for the $N=1700$ markets we exactly replicate the given number across the three values of  $\lambda$.

Our simulations of the Int-DA procedure show that it can reproduce stylized results reflective of the observed NRMP figures---over unmatched rates, over the fraction of first-ranked outcomes, and over the small cores found in rank-order list data. Moreover, the Int-DA process does so generically, across market sizes and the common-preference weights. 

Given the fit with observed data regularities, a natural question regards the difference between outcomes under Int-DA and standard DA? The final set of results in  Table~\ref{table:SimOutcomes} speak to this question. 

The \emph{identical partner to DA} row directly contrasts the Int-DA and DA match outcomes. Averaging across our parameterizations, we find that 78\% of participants in the Int-DA procedure are matched to the \emph{exact same partner} they would match to under DA with truthful preferences reports. While four of every five doctors are entirely unaffected by the interview process, 22\% of participants being affected is far from negligible.

In the last section of Table~\ref{table:SimOutcomes} we evaluate the effects on stability. For each doctor we calculate the proportion of programs they form a blocking pair with. We report the average proportion, distinguishing between matched and unmatched doctors. 
Matched doctors exhibit some instability, despite both stages in the two-stage process being chosen to select stable outcomes. Averaging across parameterizations, a blocking pair is detected for matched doctors 0.8\% of the time. Unsurprisingly, instabilities are more substantial for unmatched doctors. A randomly chosen hospital yields a blocking pair between 10\% and 33\% of the time for each unmatched doctor. 
\section{ Conclusion}
Much of the matching literature has focused on the centralized clearinghouse governing the match of newly-minted doctors and residency positions. We illustrate the possibility that \textit{decentralized} interactions preceding the match---namely, interviews---may dramatically impact ultimate outcomes.

For the NRMP, our results imply that empirical estimations based on preferences submitted to the clearinghouse should be used with great caution. More broadly, beyond the NRMP, our paper suggests that interactions outside of the clearinghouse can have dramatic effects on outcomes. 
\bibliographystyle{te}
\bibliography{./intvw-bib.bib}
\appendix
\section{Proof of Proposition 2}
\label{Proof:Prop2}

\label{sec:proofofpropLM}
A formal statement of Proposition~\ref{propLM} follows.

\begin{proposition*}
    Let $k_n\geq 1$ be a sequence of positive integers and $M\geq 1$ be a constant. 
    Let $\ep, \ta, \pi \in (0,1)$. Suppose that $\limsup k_n/n < 1$.
    Then there is $N\in \Na$ such that for all $n\geq N$ $P(E_n)\geq 1-\pi$, where $E_n$ is the set of $c_h$, $c_d$, $\eta_{d,h}$ and $\eta_{h,d}$ such that in the resulting market $(D_n,H_n,U_n)$, for a fraction of at least $1-\ta$ of doctors $d$, the rank-order of $\mu^I_n(d)$ in her interview-truncated preference is lower by at least $M$ than the rank-order of any hospital generating utility of at most $u_d(\mu^{DA}_n(d))-\ep$ in her actual preference $\succeq_d$.
\end{proposition*} 

\begin{proof}
Note that if $k'_n$ is in the hypotheses of the proposition, so is $k'_n+M$. So replace $k_n$ by $k_n+M$ in the sequel.
With some notational abuse, we drop the multipliers $\la^D$ in $1-\la^D$ and write $c_d$ for $\la^D c_d$, $\eta_{d,h}$ for $(1-\la^D)\eta_{d,h}$, etc. This re-scaling implies that utilities are sums of the common and private value components:  $u^n_d(h) = c_h + \eta_{d,h}$, and  $u^n_h(d) = c_d + \eta_{h,d}$. The relevant probability distributions are re-scaled correspondingly, but remain absolutely continuous, with support on a compact interval in $\Re$. Without loss, we assume that this interval is $[0,1]$.\footnote{In fact, the distributions do not need to have a compact support. It suffices to choose a compact set that accumulates large enough probability.
We thank SangMok Lee for this observation.}

Let $D = \cup_n D_n$ and $H=\cup_n H_n$. Consider tuples $(c,\eta)$, with $c=(c_a)_{a\in H\cup D}$ and \[ \eta = ((\eta_{a,b})_{(a,b)\in H\times D},(\eta_{a,b})_{(a,b)\in D\times H}).\] The tuples $(c,d)$ are endowed with the product probability measure from the i.i.d. distributions described above. 

Let $G$ denote the cumulative distribution function corresponding to $c_d$ and fix $\ta,\ep,\pi>0$.  


To understand how the proof works, note that if agents match assortatively based on the common component, then a doctor $d$ should be able to find a hospital $h$ for which it has idiosyncratic utility close to $1$, and this hospital should provide $d$ with (approximately) the same utility $c_d+1$ as it receives from matching with $d$. Think of $c_d+1$ as $d$'s ``target utility.''

Let \[ A_n(\ep,(c,\eta)) =  \{ d\in D_n: c_d+1 - \ep < u_d(\mu^{DA}_n(d)) < c_d + 1 + \ep \}\] be the set of doctors for which this is achieved (in DA), up to $\ep$. We shall prove that, when $n$ is large enough, with large probability, a fraction at least $1-\ta/2$ doctors are in $A_n(\ep,(c,\eta))$. 

Consider the number of hospitals ranked above a doctor's target utility $c_d+1$. Let $k_n\geq 1$ be a sequence of positive integers such that $\limsup k_n/n < 1$. Let \[B(c_d,n) =  \{ \abs{ h\in H_n : c_h + \eta_{d,h} > c_d+1 } \leq k_n \}.\] be the event that fewer than $k_n$ hospitals give $d$ a utility greater than $d$'s target utility. We denote by $\beta_n$ the probability that a fraction of at least $\theta/2$ doctors have a ``small'' number of at most $k_n$  hospitals above their target utility.


We shall prove that for $n$ large enough, $\beta_n <\pi/2$ and $P \big(\frac{1}{n} \abs{A_n(\ep,(c,\eta))} \geq 1-\ta/2\big) > 1-\pi/2$. Thus, the event that $B(c_d,n)$ is false for a fraction $\geq 1-\ta/2$ of doctors {\em and} the event $\left(\frac{1}{n} \abs{A_n(\ep,(c,\ep))} \geq 1-\ta/2\right)$ holds, has probability $\geq (1-\pi/2 ) + (1-\pi/2 ) -1 = 1-\pi$. At the intersection of these events, it holds for a fraction $\geq (1-\ta/2) + (1-\ta/2) -1 = 1-\ta$ of $d\in D_n$  that $B(c_d,n)$ is false and $d\in A_n(\ep)$. Hence, for a fraction $\geq 1-\ta$ of $d\in D_n$ there are more than $k_n$ hospitals above their target utility, and they are within $\ep$ of their target utilities.  The rank-order of any partner in $\mu^I$ is at most $k_n$, so these statements prove the proposition. 

To finish the proof we carry out the required calculations.
Let \[ l = \limsup_{n\rightarrow \infty} \frac{k_n}{n}\] and recall that $l \in [0,1)$ by hypothesis. Choose $c^\star$ and $\da>0$ such that  $1-G(c^\star) + \da < \ta/4$ and $l<P(c_h + \eta_{d,h} > c^\star+1)$. This is possible by absolute continuity of the distributions of $c_h$ and $\eta_{d,h}$. Let $p(c^\star) = P(c_h + \eta_{d,h} > c^\star+1)$. 

If $c_d\leq c^\star$, then
\begin{align}  
P( B(c_d,n)) & =  P( \sum_{h\in H_n} \one_{c_h + \eta_{d,h} > c_d+1} \leq k_n ) \notag \\
& \leq  P\left( \frac{1}{n} \sum_{h\in H_n} \one_{c_h + \eta_{d,h} > c^\star+1} \leq p(c^\star) - (p(c^\star)-\frac{k_n}{n}) \right) \notag \\
& \leq \exp (-2(p(c^\star) - \frac{k_n}{n})^2 n) \label{eq:hoeffd2}
\end{align}
by Hoeffding's inequality (observe that, eventually, $p(c^\star) - \frac{k_n}{n} > 0$).

Let
\begin{align*}
\beta_n & = P\left( \abs{ \{ d\in D_n : B(c_d,n) \} } > n\ta/2  \right) \\
& \leq  P\left( \underbrace{ \abs{\{ d\in D_n : B(d,n) \text{ and } c_d \leq  c^\star \} } }_{Z_n}+ \underbrace{ \abs{ \{ d\in D_n :  c_d > c^\star \} }}_{Y_n} > n \ta/2\right)  \\ 
& \leq  P(\frac{1}{n}Z_n + 1-G(c^\star) + \da > \ta/2) + P(\frac{1}{n}Y_n > 1-G(c^\star) + \da )
\end{align*}
The first inequality follows by counting all $d$ with $c_d> c^\star$ as if $B(d,n)$ were true. So the random variable $Y_n$ counts all $d\in D_n$ with $c_d> c^\star$ as if they were in $B(c_d,n)$.

The second inequality is a truncation exercise, partitioning the probability space into two events. The first event is  $\frac{1}{n}Y_n \leq 1-G(c^\star) + \da$ and the second is $\frac{1}{n}Y_n > 1-G(c^\star) + \da$. Under the second event, we  have  $\frac{1}{n} Z_n + \frac{1}{n} Y_n >\ta/2$ as $1-G(c^\star) +\da>\ta/2$. Under the first event, the inequality is obtained by ``raising'' $\frac{1}{n}Y_n$ to $ 1-G(c^\star) + \da$. 

Applying Hoeffding's inequality again,
\begin{equation}\label{eq:hoeffdforY}
  P (\frac{1}{n}Y_n > 1-G(c^\star) + \da) \leq \exp(-2 \da^2 n).
\end{equation}

Now, 
\begin{align}
P(Z_n  >n(\ta/2 - [1-G(c^\star) + \da)]) & 
\leq  P(\cup_{d\in D_n} B(d,n)) | c_d = c^\star) \notag \\
& \leq \sum_{d\in D_n} P(B(d,n) | c_d = c^\star) \notag \\
& \leq n \exp(-2 (p(c^\star) - \frac{k_n}{n})^2n ), \label{eq:hoeffforZ}
\end{align} 
where the first inequality follows as $n(\ta/2 - (1-G(c^\star) + \da))\geq 1$, and the probability of $B(d,n)$ is maximized by setting $c_d=c^\star$.

Choose $n$ such that 
\begin{align}
n(\ta/2 - [1-G(c^\star) + \da])& > 1, \label{eq:one} \\
\exp(-2 \da^2 n) & < \pi /4, \label{eq:two} \\ 
n\exp (-2(p(c^\star) - \frac{k_n}{n})^2 n) & <\pi/4, \label{eq:three} \\
\text{ and } P \left(\frac{1}{n} \abs{A_n(\ep,(c,\eta))} \geq 1-\ta/2\right) & > 1-\pi/2. \label{eq:SM}
\end{align} 
Observe that \eqref{eq:one} is possible as $\ta/2 - [1-G(c^\star) + \da]>0$. Inequality~\eqref{eq:three} requires that $k$ is $O(n)$, which holds by hypothesis,  and our choice of $c^\star$ to ensure that $p(c^\star) - k_n/n>0$ is eventually bounded below.  Inequality~\eqref{eq:SM} is possible by Theorem~1 of  \cite{lee2016incentive}. 

By \eqref{eq:hoeffdforY},\eqref{eq:hoeffforZ},\eqref{eq:two}, and \eqref{eq:three}, we obtain that 
\begin{equation}
\beta_n \leq  n \exp(-2 (p(c^\star) - \frac{k}{n})^2n ) + \exp(-2 \da^2 n) < \pi/2 \label{eq:bringithome2}
\end{equation}

Statements~\eqref{eq:SM} and \eqref{eq:bringithome2} provide the two bounds needed.
\end{proof}

\pagebreak

\setcounter{page}{1}

\begin{centering}
    {\Large \bf For Online Publication--Appendix: Omitted Proofs}    
\end{centering}

\subsection{Proposition 1}
\begin{proof}
Let $\succeq$ be the common preference that doctors have over hospitals. Note that DA is the same as serial dictatorship (SD) with the order dictated by hospital rank in $\succeq$. 

Consider a doctor $d$ assigned to $h=\mu^{DA}(d)$ in the $r$th round of SD. The rank-order of $h$ in $d$'s preference is therefore $r$. If $k\leq r$ then we are done, as the rank-order of $\mu^I(d)$ in $d$'s truncated preference is at most $k$.  

Suppose that $r< k$. Two observations follow. First, consider the interview stage and a hospital $h=\mu^{DA}(d)$ matched to $d$ in stage $r'< k$ of DA. When choosing whom to interview, $h$ can choose any doctor, as all of them would have received strictly fewer than $k$ interview requests when they get a request from $h$. So the hospital choosing at stage $r'$ of DA will interview the highest $k$ doctors in her preference. 

Second,  $\mu^{DA}(h) = \mu^{I}(h)$ for the hospital $h$ choosing at round $r$.\footnote{Incidentally this may not happen for hospitals choosing at round $r'>k$. It is easy to come up with examples.} This is shown by induction: The statement is obviously true for the highest ranked hospital. Suppose that $\mu^{DA}(h) = \mu^{I}(h)$ for all hospitals choosing at any stage $r'<r$. If $h$ is the $r$-ranked hospital then the set of doctors available to $h$ in the DA stage of Int-DA is $D$, by our first observation, minus the choices of hospitals with rank-order $r'<r$. By the inductive hypothesis the doctors chosen by the hospitals with rank-order $r'<r$ is the same as DA. So the set of available doctors to hospital $h$ is the same in Int-DA as in DA. Thus $\mu^{DA}(h)=\mu^{I}(h)$. 
\end{proof}

\subsection{Proposition~3}
\begin{proof}
Specifically, we show that there are constants $N$, $K$, $K'$, $K''$ and $K'''$ that do not depend on $\ta$ and $\pi$, such that for all  
\[  n\geq \max\{ 
\bar N,   \frac{\ln(\pi/4)}{K}, \frac{\ln(4/\pi)}{2\da^2}, (\frac{\ta}{2} + K')^{-1}, (\frac{12}{\ta})^{4},\left(\frac{\log (1-\frac{\pi}{2})}{\log K''} + 3\right)^4K''',
\} 
\]
the statement in Proposition~\ref{propLM} holds.

The market size in the proof of Proposition~\ref{propLM} is determined from inequalities~\eqref{eq:one}-\eqref{eq:SM}. These are the starting point of the proof. Using the bounds in \cite{lee2016incentive}, these mean that we need to choose  $n$ such that 
\begin{align}
-2 [p(c^\star) - \frac{k_n}{n}]^2 n & \leq \ln (\frac{\pi}{4n})  \label{eq:threebounds} \\
-2 \da^2 n & < \ln(\frac{\pi}{4}), \label{eq:twobounds} \\ 
\frac{1}{\frac{\ta}{2} - (1-G(c^\star) + \da )}  & < n \label{eq:onebounds} \\
\frac{2}{n}\left(\frac{1}{n^{1/4}} - 3 \right) \sqrt{n}\log (n) + \frac{6}{n^{1/4}} & > \frac{\ta}{2} \label{eq:fourbounds} \\
(1- g_n)^{2n^{1/4} - 4 } & \geq 1-\frac{\pi}{2} \label{eq:fivebounds} ,
\end{align}
where $g_n$ is $o(e^{-\sqrt{n}\log n})$

For \eqref{eq:threebounds}, choose $N_0$ and $K_0$ such that if $n\geq N_0$ then $(p(c^\star)-k_n/n)^2 \leq K_0$. This is possible given the hypothesis that $\limsup k_n/n < 1$. Next, let $N_1\geq N_0$ and $K_1$ be such that, for all $n\geq N_1$, $2K_0 n - \ln n \geq K_1 n$. Then we need that 
\begin{equation}
K_2 n \geq \ln(\frac{4}{\pi}) \label{eq:threefinal}    
\end{equation}

For \eqref{eq:twobounds} and \eqref{eq:onebounds}, we have  
\begin{align} 
    n & > \frac{\ln(4/\pi)}{2\da^2} \label{eq:twoboundsfinal}\\
    n & \geq \frac{1}{\frac{\ta}{2} - (1-G(c^\star)) -\da} \label{eq:oneboundsfinal}
\end{align}

For \eqref{eq:fourbounds} we need that 
\begin{align*}
    \frac{2 \sqrt{n} \log n}{n^{5/4}} - \frac{6 \sqrt{n} \log n}{n} + \frac{6}{n^{1/4}} & < \frac{\ta}{2} \\ \Longleftrightarrow 
    \frac{2 \log n}{n^{1/4}}\left(\frac{1}{\sqrt{n}} - \frac{3}{n^{1/4}} \right) + \frac{6}{n^{1/4}} & < \frac{\ta}{2} \\
\end{align*}
Let $N_2\geq N_1$ be such that for all $n\geq N_2$, $\frac{1}{\sqrt{n}} - \frac{3}{n^{1/4}} \leq 0$. Then all we need is that  $\frac{6}{n^{1/4}} < \frac{\ta}{2}$, or that 
\begin{equation}
n\geq (\frac{12}{\ta})^{4}.\label{eq:fourboundsfinal}
\end{equation}

For \eqref{eq:fivebounds}, fix $N_3\geq N_2$ and $K_4$ such that for all $n\geq N_3$ $1-g_n\geq K_4$. So we need to obtain $\log(1-\frac{\pi}{2}))\leq (2n^{1/4} - 3)\log K_3$. That is,
\begin{equation}n\geq \left(\frac{\log (1-\frac{\pi}{2})}{\log K_3} + 3\right)^4\frac{1}{16} \label{eq:fiveboundsfinal}    
\end{equation}
Set $\bar N = N_3$, $K=K_2$ $K'=(1-G(c^\star)) +\da$, $K'' = $ $K'''=1/16$. Then the calculations above correspond to~\eqref{eq:threefinal},~\eqref{eq:twoboundsfinal},~\eqref{eq:oneboundsfinal},~\eqref{eq:fourboundsfinal}, and~\eqref{eq:fiveboundsfinal}.
\end{proof}

\section{For Online Publication--Appendix: Additional Figures}
\setcounter{figure}{0}
\renewcommand{\thefigure}{B.\arabic{figure}}
\begin{figure}[h]
\centering
\includegraphics[width=0.7\textwidth]{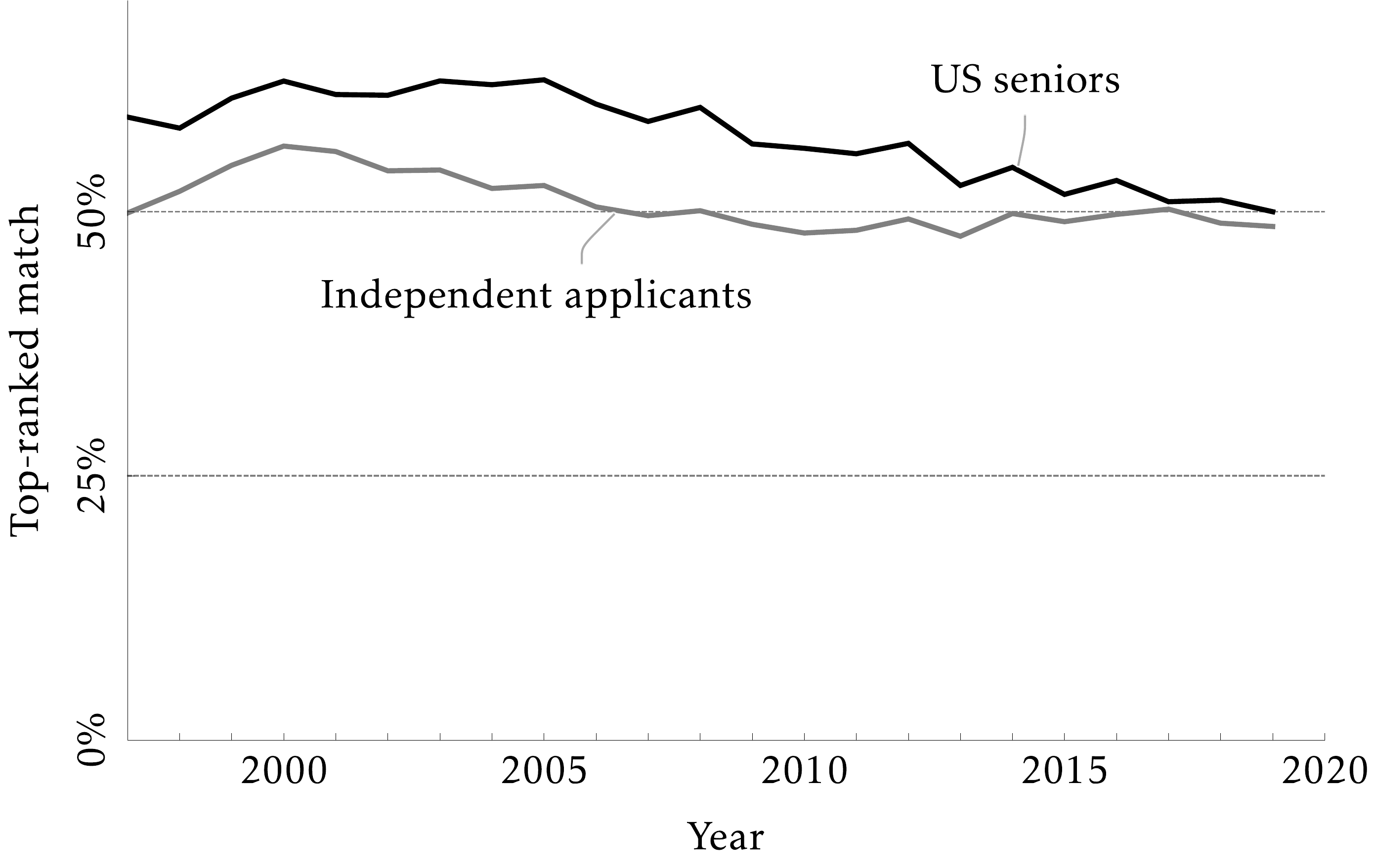}

\caption{NRMP residents matched to first-ranked program (conditional on matching)}
\label{fig:setuppuzzle}
\end{figure}
\end{document}